\documentclass[11pt,a4paper,titlepage]{article}
\usepackage[utf8x]{inputenc}
\usepackage[T2A]{fontenc}
\usepackage[english]{babel}
\usepackage{mathrsfs}
\usepackage{ucs}
\usepackage{amsmath}
\usepackage{amsfonts}
\usepackage{amssymb}
\usepackage{amsthm}
\usepackage{longtable}
\usepackage{enumerate}
\usepackage{bussproofs}
\usepackage{proof}
\usepackage{xcolor}
\usepackage[width=0.00cm, height=0.00cm, left=2.00cm, right=2.00cm, top=2.00cm, bottom=2.00cm]{geometry}
\usepackage{url}

\theoremstyle{definition}
\newtheorem{defn}{Definition}[section]
\newtheorem{theorem}{\textbf{Theorem}}
\newtheorem{lemma}{\textbf{Lemma}}

\newcommand{\1}{\vspace{.1in}}

\newcommand{\Ex}{E}
\usepackage{units}
 \newcommand{\half}{\nicefrac{1}{2}}
  \newcommand{\val}{\mathcal{V}}
\newcommand{\bc}{\begin{center}}
	\newcommand{\ec}{\end{center}}
\usepackage{bussproofs,longtable}
\usepackage[breaklinks,colorlinks,citecolor=blue,urlcolor=blue]{hyperref}

\newenvironment{Published}{\begin{list}{}{\leftmargin=25pt}
    \item[]\small\noindent{To be published as}}%
    {\end{list}}
    \newenvironment{Funding}{\begin{list}{}{\leftmargin=25pt}
        \item[]\small\noindent{The research}}%
        {\end{list}}
\newenvironment{Abstract}{\begin{list}{}{\leftmargin=25pt}
    \item[]\footnotesize\noindent{\sc \bf Abstract:}}%
{\end{list}}
\newenvironment{keywords}{\begin{list}{}{\leftmargin=25pt}
    \item[]\footnotesize\noindent{\sc \bf Keywords:}}%
{\end{list}}

\begin{document}
\begin{center}
\vspace{0.3cm}
{\large Andrzej Indrzejczak, 
Yaroslav Petrukhin\vspace{0.3cm}\\
\textbf{Bisequent Calculi for Neutral Free Logic with Definite Descriptions}}

\ \\
\begin{footnotesize}
\noindent Department of Logic, University of Lodz, Poland,  \url{andrzej.indrzejczak@filhist.uni.lodz.pl},\\
\noindent Center for Philosophy of Nature, University of Lodz, Poland,  \url{yaroslav.petrukhin@gmail.com}\\
\end{footnotesize}
\end{center}
\begin{Published}
 Indrzejczak, A., Petrukhin, Y. (2024). Bisequent Calculi for Neutral Free Logic with Definite Descriptions. 
ARQNL 2024: Automated Reasoning in Quantified
  Non-Classical Logics, 1 July 2024, Nancy, France. 
%\DOI{http://dx.doi.org/}{ }
\bigskip
\end{Published}
\begin{Funding}
in this paper was funded by the European Union
 (ERC, ExtenDD, project number: 101054714). Views and opinions expressed
 are however those of the author(s) only and do not necessarily reflect those of
 the European Union or the European Research Council. Neither the European
 Union nor the granting authority can be held responsible for them.
\end{Funding}
 \begin{Abstract}
	We present a bisequent calculus (BSC) for the minimal theory of definite descriptions (DD) in the setting of neutral free logic, where formulae with non-denoting terms have no truth value. The treatment of quantifiers, atomic formulae and simple terms is based on the approach developed by Pavlovi\'{c} and Gratzl. We extend their results to the version with identity and definite descriptions. In particular, the admissibility of cut is proven for this extended system.
 \end{Abstract}

\section{Introduction}

There is a variety of logics called free logics (FL) and differing in many respects. The common feature is that singular terms are free from existential assumptions, i.e. they are not assumed to denote an existing object. On the other hand, in all free logics quantifiers are assumed to have an existential import. One of the main divisions of this kind of logics is based on the treatment of atomic formulae with non-denoting terms. In positive FL they can be true, in negative FL they are always false, in neutral FL (NFL) they are neither true nor false. This makes NFL in a sense a kind of partial or three-valued logic, with the third value expressing the truth-value gap.

We focus on the proof-theoretic perspective on free logics, namely the examination of sequent calculi (SC). In case of positive and negative FL one may find several studies, due to Maffezioli and Orlandelli \cite{MaffezioliOrlandelli}, Pavlovi\'{c} and Gratzl \cite{PavlovicGratzl21}, or Indrzejczak \cite{IndrzejczakFreeLogic}. However, in case of NFL the situation is worse in general, and particularly poor in the field of proof theory. In addition to the papers of Woodruff \cite{Woodruff} and Lehmann \cite{Lehmann}, where one may find some kind of natural deduction and tableau system, only Pavlovi\'{c} and Gratzl \cite{PavlovicGratzl23} recently investigated SC for NFL based on strong and weak Kleene logics \cite{Kleene}. 

We are going to continue the research carried out in \cite{PavlovicGratzl23} in two ways: first, we reformulate the results from \cite{PavlovicGratzl23} using our own approach based on bisequent calculi (BSC), developed in \cite{IndrzejczakPetrukhin,Indrzejczak3}; second, we extend the resulting proof systems to cover identity and definite descriptions (DD). 
These complex terms are usually divided into proper (denoting a unique object, like `the present King of England') and improper (like non-denoting `the present King of France' or not unique `the author of Principia Mathematica'). Improper DD are very troublesome and it seems that NFL is quite natural place for exploring their behaviour. Usually, DD are formalised by means of term-forming $ \imath $-operator applied to a variable $ x $ and formula $ \varphi $ to obtain a term $ \imath x \varphi $. 
It is worth noting that there is an alternative treatment of DD based on the application of binary quantifier and recently developed by K\"{u}rbis \cite{Kurbis1,Kurbis2,Kurbis3}. We follow here a more standard, $ \imath $-operator based approach.

The theories of DD in positive or negative free logics are usually based on Lambert axiom \eqref{L}:
\begin{align}
\label{L}\tag{$L$} 
\forall y (y=\imath x \varphi  \leftrightarrow \forall x (\varphi\leftrightarrow y=x))
\end{align}

\noindent where $\imath x\varphi$ is closed and does not have any occurrence of $y$.
It is the weakest  principle characterising the behaviour of proper DD.  
In our research, devoted to NFL, we consider \eqref{L} as well but
to avoid using $\leftrightarrow$, which considerably complicate the set of required rules and proofs in the setting of Kleene's logic, we will use instead two axioms which together express the same minimal theory of DD:
\begin{align}
\label{L1}\tag{$L^\rightarrow$} 
\forall y (y=\imath x \varphi  \rightarrow \varphi[x/y] \wedge \forall x (\varphi\rightarrow y=x))\\
\label{L2}\tag{$L^\leftarrow$} \forall y (\varphi[x/y]\wedge\forall x(\varphi\rightarrow y=x) \rightarrow y=\imath x\varphi)
\end{align}

Moreover, we add a restriction:  $\imath x\varphi$ contains no other DD inside. There is a possible treatment of this theory admitting nested DD but it leads to the formulation of rules which do not allow us to prove the admissibility of cut. 

Various definitions of the logical connectives and quantifiers can be considered in many-valued logics in general, and the choice of them has an influence also on the behaviour of DD. In this study, following \cite{PavlovicGratzl23}, we examine two variants of NFL, based on strong and weak Kleene's \cite{Kleene} interpretation of connectives. 

The goal of the paper is to provide a uniform  proof-theoretic treatment of NFL with identity and DD. Regarding the systems for DD in the setting of other logics, recently a great progress has been made in the works of Fitting and Mendelsohn \cite{fitt:fir98}, Orlandelli \cite{Orlandelli2021} and Indrzejczak \cite{IndrzejczakRussellDD,IndZaw2,IndNils}. The paper \cite{IndrzejczakFreeDD}, particularly important for us, is devoted to sequent calculi for the theories of DD based on positive and negative free logics. Currently we are going to continue this research and examine the behaviour of DD in NFL, however this logic requires some nonstandard sequent basis, like the one originally applied in \cite{PavlovicGratzl23}. We prefer to use for this aim a slightly different framework, namely a bisequent calculus (BSC) \cite{IndrzejczakPetrukhin,Indrzejczak3}, since it has shown its usefulness in the field of formalisation of three- and four-valued logics. Moreover, the rules of the generalised sequent calculus from \cite{PavlovicGratzl23} are easily translatable into BSC.
Summing up: the present paper is a combination and a continuation of the research presented in \cite{PavlovicGratzl21,IndrzejczakFreeLogic} (SC for negative and positive free logics), \cite{PavlovicGratzl23} (SC for neutral free logics), \cite{IndrzejczakFreeDD} (SC for DD theories based on negative and positive free logics), and \cite{IndrzejczakPetrukhin,Indrzejczak3} (BSC for many-valued logics).

The structure of the paper is as follows. Section \ref{Logics} contains some preliminaries related to the language and semantics. Section \ref{StrongKleene} is devoted to the introduction of the BSC for two variants of NFL based on strong and weak Kleene's logic. 
Section \ref{Cut} is concerned with the constructive proof of cut admissibility.
Section \ref{Conclusion} consists of some concluding remarks and open problems.

\section{Preliminaries}\label{Logics}
Following \cite{PavlovicGratzl23}, we consider the subsequent language but with added conjunction, for dealing with \eqref{L}.
\begin{defn}\cite[Definition 1.1]{PavlovicGratzl23}
	The alphabet of the language $ \mathcal{L} $ consists of:
	\begin{enumerate}[$ (1) $]\itemsep=0pt
		\item Terms: $ t, s, \ldots, $ consisting of:
		\begin{enumerate}[$ (a) $]\itemsep=0pt
			\item Denumerable list of free individual variables \textup{(}parameters\textup{)}: $  a,b, c, \ldots $
			\item Denumerable list of bound individual variables: $ x, y, z,  \ldots $
		\end{enumerate}
		\item Denumerable list of $ n $-ary predicates,  including a unary predicate $ E $
		\item $ \neg,\rightarrow, \wedge, \forall, (, ) $.
	\end{enumerate}
\end{defn}

For our purposes we extend the language $ \mathcal{L} $ with identity and $ \imath $-operator.
\begin{defn}
	The alphabet of the language $ \mathcal{L}_\imath^= $ consists of all symbols of the alphabet of the language $ \mathcal{L} $ as well as $ = $ and $ \imath $.
\end{defn}

The notions of terms and formulae of the languages $ \mathcal{L} $ and $ \mathcal{L}_\imath^= $ are defined by simultaneous recursion. Note that in $ \mathcal{L}_\imath^= $, $t, s$ represent arbitrary terms, including DD. Complexity of formulae and terms is measured by the number of logical symbols (including predicates $E$ and =).
$\varphi[t_1/t_2] $ is used for the operation of correct substitution of an arbitrary term
$ t_2 $ for all occurrences of a variable/parameter $ t_1 $ in $ \varphi $, and similarly
$ \Gamma[t_1/t_2] $ for a uniform substitution in all formulae in $ \Gamma $. 
We use a simplified semantics from \cite{PavlovicGratzl23,PavlovicGratzl21} to provide interpretation of this language.

\begin{defn}[Neutral structure $ \mathcal{S}_{nt} $]\cite[Definition 4.1]{PavlovicGratzl23} \cite[Definition 3.1]{PavlovicGratzl21}
	A neutral structure $ \mathcal{S}_{nt} $ is a pair $ \langle \mathcal{D}, \mathcal{I}\rangle $,
	where $\mathcal{D} = a_1, \ldots , b_1, \ldots$ is a countable list of parameters, and $ \mathcal{I} $ is an
	interpretation function on $ \mathcal{L} $ \textup{(}$ \mathcal{L}_\imath^= $\textup{)}:
	\begin{itemize}\itemsep=0pt
		\item $\mathcal{I}(t) = t $, where $ t\in \mathcal{D} $,
		\item $\mathcal{I}(\Ex)\subseteq\mathcal{D} $,
		\item $\mathcal{I}(=)=Ref \cup Id $, closed under symmetry and transitivity, where 
		\begin{itemize}\itemsep=0pt
			\item 
			$ Ref = \{\langle t,t\rangle \mid t \in \mathcal{I}(\Ex)\} $,
			\item 
			$ Id \subseteq \mathcal{I}(\Ex) \times \mathcal{I}(\Ex) $, 		
		\end{itemize}
		\item $ \mathcal{I}(P^n) \subseteq \mathcal{I}(\Ex)^n $ such that if $ \langle s, t\rangle \in \mathcal{I}(=) $, then $ \langle \ldots , s_i, \ldots\rangle \in \mathcal{I}(P^n)$ iff
		$ \langle\ldots , t_i, \ldots\rangle\in \mathcal{I}(P^n) $, for any $ n $ and any $ 1 \leqslant i \leqslant n $.
	\end{itemize}
\end{defn}

Note that identity is in principle treated as other predicates but it cannot be defined in this way since domains of models contain just parameters, so it should be defined as a condition on models like in \cite{PavlovicGratzl21}. 
$\mathcal{I}$ is extended to cover descriptions $\imath x\varphi$, where $\varphi$ does not contain other DD, by partial mapping from the set of so restricted DDs to $\mathcal{D}$. In case $\mathcal{I}(\imath x\varphi)=a$, for some $a\in \mathcal{I}(\Ex)$, we say that $\imath x\varphi$ is defined and assume that it satisfies the following condition: 
\[\mathcal{I}(\imath x \varphi)= 
\begin{cases}
\mbox{$a$ iff $\mathcal{V}^i(\varphi[x/a])=1$ and $\mathcal{V}^i(\varphi[x/b])=0$ for every $a\neq b\in \mathcal{I}(\Ex)$;} \\
\mbox{otherwise it is undefined}.
\end{cases}\]

There is a matter of debate (cf. \cite{Lehmann}) which Kleene's logic, strong or weak, provides better option for interpreting truth-value gaps in NFL. Pavlovi\'{c} and Gratzl \cite{PavlovicGratzl23} consider both options and we follow their approach. Accordingly
$\mathcal{V}^i$ may be either weak ($i = w$) or strong ($i = s$), and is defined as follows:
\begin{defn}[Weak valuation $ \mathcal{V}^w $]\textup{(}An extended version of \cite[Definition 4.2]{PavlovicGratzl23}\textup{)}\label{WeakVal}
	The truth-value assignment $ \mathcal{V}^w $ on the structure
	$ \langle \mathcal{D},\mathcal{I}\rangle $ is defined as follows:
	\begin{align} 
	\mathcal{V}^w(\Ex t) &=\left\{\begin{array}{cll}
	1&\hbox{~iff~}t\in\mathcal{I}(\Ex),\\
	0&\hbox{~iff~otherwise};\\
	\end{array}\right.\\
	\mathcal{V}^w(P^n(t_1,\ldots,t_n))&=\left\{\begin{array}{cll}
	1&\hbox{~iff~}\langle t_1,\ldots,t_n\rangle\in \mathcal{I}(P^n),\\
	\half & \hbox{~iff~for~some~}1\leqslant i\leqslant n,t_i\not\in \mathcal{I}(\Ex),\\
	0&\hbox{~iff~otherwise};\\
	\end{array}\right.\\
	\mathcal{V}^w(\neg \varphi)&=\left\{\begin{array}{cll}
	1&\hbox{~iff~}\mathcal{V}^w(\varphi)=0,\\
	\half &\hbox{~iff~}\mathcal{V}^w(\varphi)=\half ,\\
	0&\hbox{~iff~}\mathcal{V}^w(\varphi)=1;\\
	\end{array}\right. \\
	\mathcal{V}^w(\varphi\wedge \psi)&=\left\{\begin{array}{cll}
	1&\hbox{~iff~}\mathcal{V}^w( \varphi)=1\hbox{~and~}\mathcal{V}^w(\psi)=1,\\
	\half &\hbox{~iff~}\mathcal{V}^w( \varphi)=\half\hbox{~or~}\mathcal{V}^w( \psi)=\half ,\\
	0&\hbox{~iff~otherwise};\\
	\end{array}\right.  \\		
	\mathcal{V}^w(\varphi\rightarrow \psi)&=\left\{\begin{array}{cll}
	0&\hbox{~iff~}\mathcal{V}^w( \varphi)=1\hbox{~and~}\mathcal{V}^w( \psi)=0,\\
	\half &\hbox{~iff~}\mathcal{V}^w( \varphi)=\half\hbox{~or~}\mathcal{V}^w( \psi)=\half ,\\
	1&\hbox{~iff~otherwise};\\
	\end{array}\right.  \\
	\mathcal{V}^w(\forall x \varphi)&=\left\{\begin{array}{cll}
	1&\hbox{~iff~for~every~}t\in\mathcal{I}(\Ex),\mathcal{V}^w(\varphi[x/t])=1,\\
	0&\hbox{~iff~for~some~}t\in\mathcal{I}(\Ex),\mathcal{V}^w(\varphi[x/t])=0,\\
	\half&\hbox{~iff~otherwise}.\\
	\end{array}\right.		   		     
	\end{align}
\end{defn}

\begin{defn}[Strong valuation $ \mathcal{V}^s $]\cite[Definition 4.3]{PavlovicGratzl23}
	The truth-value assignment $ \mathcal{V}^s $ on the structure
	$ \langle \mathcal{D},\mathcal{I}\rangle $ is defined as follows \textup{(}the other cases are defined as in Definition \ref{WeakVal}\textup{)}:
	\begin{align} 
	\tag{$ 4^\prime $} \mathcal{V}^s(\varphi\wedge \psi)=\left\{\begin{array}{cll}
	0&\hbox{~iff~}\mathcal{V}^s(\varphi)=0\hbox{~or~}\mathcal{V}^s(\psi)=0,\\
	1 &\hbox{~iff~}\mathcal{V}^s(\varphi)=1\hbox{~and~}\mathcal{V}^s(\psi)=1 ,\\
	\half&\hbox{~iff~otherwise}.\\
	\end{array}\right.\\
	\tag{$ 5^\prime $} \mathcal{V}^s(\varphi\rightarrow \psi)=\left\{\begin{array}{cll}
	0&\hbox{~iff~}\mathcal{V}^s(\varphi)=1\hbox{~and~}\mathcal{V}^s(\psi)=0,\\
	1 &\hbox{~iff~}\mathcal{V}^s(\varphi)=0\hbox{~or~}\mathcal{V}^s(\psi)=1 ,\\
	\half&\hbox{~iff~otherwise}.\\
	\end{array}\right. 
	\end{align}
\end{defn}

For each of these valuations, we can obtain two different logics by taking the set of designated values $D$ either as $\{ 1 \}$ or $\{1, \half \}$. However, we examine here only the first route, hence in what follows $D$ always denotes $\{ 1 \}$. 
The notion of the entailment (consequence) relation is defined as follows.
\begin{defn}
	For any set of formulas $ \Gamma $ and any formula $ \varphi $, it holds that 
	\begin{center}
		$\Gamma \models\varphi $ iff for any valuation $\mathcal{V}$:
		if $\mathcal{V}(\Gamma)\subseteq D$, then $\mathcal{V}(\varphi)\in D$, where $ D = \{1\} $. 
	\end{center}
\end{defn}

In what follows the strong Kleene's logic will be referred to as $ \bf K_3$ and the weak one as $ \bf K_3^w$.
Although propositional $ \bf K_3$ has an empty set of validities it is not so for related NFL as defined by the semantics above; in particular it holds:

\begin{lemma}\label{LambertAxiomIsValid}
	\eqref{L1} and \eqref{L2} are valid in $ \bf K_3$ and $ \bf K_3^w$.
\end{lemma}
\begin{proof}
 As an example, let us prove that \eqref{L1} is valid in $ \bf K_3$. 
Suppose that \eqref{L1} is not valid in $ \bf K_3$, that is $ \val^s\big(\eqref{L1}\big)\not=1 $. Suppose that $ \val^s\big(\eqref{L1}\big)=0 $.  Then for some $t\in\mathcal{I}(\Ex)$, $ \val^s\big(t=\imath x \varphi  \rightarrow \varphi[x/t] \wedge \forall x (\varphi\rightarrow t=x)\big)=0$. Consequently,  $ \val^s\big(t=\imath x \varphi\big)=1$ and $  \val^s\big( \varphi[x/t] \wedge \forall x (\varphi\rightarrow t=x)\big)=0$. Therefore, $ \val^s\big( \varphi[x/t]\big)=0$ or $ \val^s\big(\forall x (\varphi\rightarrow t=x)\big)=0 $. Since $t\in\mathcal{I}(\Ex)$ and $ \val^s\big(t=\imath x \varphi\big)=1$, $\mathcal{V}^s(\varphi[x/t])=1$ and $\mathcal{V}^s(\varphi[x/b])=0$ for every $t\neq b\in \mathcal{I}(\Ex)$. Since $\mathcal{V}^s(\varphi[x/t])=1$, we get $ \val^s\big(\forall x (\varphi\rightarrow t=x)\big)=0 $. Then for some $u\in\mathcal{I}(\Ex)$, $ \val^s\big(\varphi[x/u]\rightarrow t=u\big)=0$. Hence, $ \val^s\big(\varphi[x/u]\big)=1 $ and $ \val^s\big(t=u\big)=0 $. Since $t\neq u\in \mathcal{I}(\Ex)$, $\mathcal{V}^s(\varphi[x/u])=0$. A contradiction. 
	
	Suppose that $ \val^s\big(\eqref{L1}\big)=\half $. Let $ (H^\rightarrow) $ denote $ t=\imath x \varphi  \rightarrow \varphi[x/t] \wedge \forall x (\varphi\rightarrow t=x) $. Hence, we can infer that either there is $t\not\in\mathcal{I}(\Ex)$ and $ \val^s\big((H^\rightarrow)\big)\in\{1,\half,0 \}$ or there is $t\in\mathcal{I}(\Ex)$ such that $ \val^s\big((H^\rightarrow)\big)=\half$. Suppose that there is $t\not\in\mathcal{I}(\Ex)$ and $ \val^s\big((H^\rightarrow)\big)\in\{1,\half,0 \}$. However, $\mathcal{I}(=)=Ref \cup Id $,  where $ Ref = \{\langle t,t\rangle \mid t \in \mathcal{I}(\Ex)\} $ and
	$ Id \subseteq \mathcal{I}(\Ex) \times \mathcal{I}(\Ex) $. Thus, it cannot be the case that there is $t\not\in\mathcal{I}(\Ex)$ and $ \val^s\big((H^\rightarrow)\big)\in\{1,\half,0 \}$. Hence, there is $t\in\mathcal{I}(\Ex)$ such that $ \val^s\big((H^\rightarrow)\big)=\half$. It cannot be the case $\val^s\big(t=\imath x\varphi\big)=\half$, otherwise $t\not\in\mathcal{I}(\Ex)$ or $\imath x\varphi\not\in\mathcal{I}(\Ex)$, but $ \mathcal{I}(=)\subseteq\mathcal{I}(E)$.  Thus, $\val^s\big(t=\imath x\varphi\big)=1$ and $ \val^s\big(\varphi[x/t] \wedge \forall x (\varphi\rightarrow t=x)\big)=\half $. Since $t\in\mathcal{I}(\Ex)$ and $ \val^s\big(t=\imath x \varphi\big)=1$, $\mathcal{V}^s(\varphi[x/t])=1$ and $\mathcal{V}^s(\varphi[x/b])=0$ for every $t\neq b\in \mathcal{I}(\Ex)$. Since $\mathcal{V}^s(\varphi[x/t])=1$ and $ \val^s\big(\varphi[x/t] \wedge \forall x (\varphi\rightarrow t=x)\big)=\half $, we have 
	$  \val^s\big(\forall x (\varphi\rightarrow t=x)\big)=\half $. Hence, we can infer that either there is $u\not\in\mathcal{I}(\Ex)$ and $ \val^s\big(\varphi[x/u]\rightarrow t=u\big)\in\{1,\half,0 \}$ or there is $u\in\mathcal{I}(\Ex)$ such that $ \val^s\big(\varphi[x/u]\rightarrow t=u\big)=\half$. Due to the properties of $ \mathcal{I}(=) $, only the latter option holds. Since $ \mathcal{I}(=)\subseteq\mathcal{I}(E)$, $ \val^s\big(t=u\big)\not=\half $. Thus, since $ \val^s\big(\varphi[x/u]\rightarrow t=u\big)=\half$, we obtain $ \val^s\big(\varphi[x/u]\big)=\half$ and $ \val^s\big(t=u\big)=0$. Since $t,u\in\mathcal{I}(\Ex)$, $t\not=u$, and $\mathcal{V}^s(\varphi[x/b])=0$ for every $t\neq b\in \mathcal{I}(\Ex)$, we infer that $\mathcal{V}^s(\varphi[x/u])=0$. But we already have that $ \val^s\big(\varphi[x/u]\big)=\half$. A contradiction.    
\end{proof}

\section{Bisequent Calculi for NFL}\label{StrongKleene}

\begin{figure}[t!]
	\centering
	\bgroup

	\noindent\begin{tabular}{ll}

		$(\mathord{\neg}\mathord{\Rightarrow\mid})$ \ $\dfrac{\Gamma
			\Rightarrow \Delta \mid \Pi\Rightarrow\Sigma, \varphi}{\neg\varphi, \Gamma \Rightarrow \Delta \mid \Pi\Rightarrow\Sigma}$ &
		$(\mathord{\Rightarrow}\mathord{\neg\mid})$ \ $\dfrac{\Gamma
			\Rightarrow \Delta \mid \varphi, \Pi\Rightarrow\Sigma}{\Gamma
			\Rightarrow \Delta, \neg\varphi \mid \Pi\Rightarrow\Sigma}$\\[16pt]
		
		$(\mathord{\mid\neg}\mathord{\Rightarrow})$ \ $\dfrac{\Gamma
			\Rightarrow \Delta, \varphi \mid \Pi\Rightarrow\Sigma}{\Gamma \Rightarrow \Delta \mid \neg\varphi, \Pi\Rightarrow\Sigma}$ &
		$(\mathord{\mid\Rightarrow}\mathord{\neg})$ \ $\dfrac{\varphi, \Gamma
			\Rightarrow \Delta \mid \Pi\Rightarrow\Sigma}{\Gamma
			\Rightarrow \Delta \mid \Pi\Rightarrow\Sigma, \neg\varphi}$\\[16pt]
		$(\mathord{\wedge}\mathord{\Rightarrow\mid})$ \ $\dfrac{\varphi, \psi,
			\Gamma \Rightarrow \Delta \mid S}{\varphi\wedge\psi, \Gamma
			\Rightarrow \Delta \mid S}$ &
		$(\mathord{\Rightarrow}\mathord{\wedge\mid})$ \ $\dfrac{\Gamma
			\Rightarrow \Delta, \varphi \mid S \qquad
			\Gamma\Rightarrow \Delta,
			\psi \mid S}{\Gamma \Rightarrow \Delta,
			\varphi\wedge\psi \mid S}$\\[16pt]
		
		$(\mathord{\mid\wedge}\mathord{\Rightarrow})$ \ $\dfrac{S \mid \varphi, \psi,
			\Gamma \Rightarrow \Delta}{S \mid \varphi\wedge\psi, \Gamma
			\Rightarrow \Delta}$ &
		$(\mathord{\mid\Rightarrow}\mathord{\wedge})$ \ $\dfrac{S\mid \Gamma
			\Rightarrow \Delta, \varphi \qquad
			S\mid \Gamma\Rightarrow \Delta,
			\psi}{S\mid\Gamma \Rightarrow \Delta,
			\varphi\wedge\psi }$\\[16pt]
		$(\mathord{\Rightarrow}\mathord{\rightarrow}\mid)$ \, $\dfrac{\Gamma
			\Rightarrow \Delta, \psi \mid \varphi, \Pi\Rightarrow\Sigma}{\Gamma\Rightarrow\Delta, \varphi\rightarrow\psi \mid \Pi\Rightarrow\Sigma}$ &
		$(\mathord{\mid\Rightarrow}\mathord{\rightarrow})$ \, $\dfrac{\varphi, \Gamma
			\Rightarrow \Delta \mid\Pi\Rightarrow\Sigma, \psi}{\Gamma\Rightarrow\Delta \mid \Pi\Rightarrow\Sigma, \varphi\rightarrow\psi}$	
	\end{tabular}
	
	\1
	
	$(\mathord{\rightarrow}\mathord{\Rightarrow}\mid)$ \, $\dfrac{\Gamma
		\Rightarrow \Delta \mid \Pi\Rightarrow\Sigma, \varphi \quad\psi, \Gamma\Rightarrow\Delta \mid
		\Pi \Rightarrow \Sigma } {\varphi\rightarrow\psi, \Gamma
		\Rightarrow \Delta \mid \Pi\Rightarrow\Sigma}$
	\1
	
	$(\mathord{\mid\rightarrow}\mathord{\Rightarrow})$ \, $\dfrac{\Gamma
		\Rightarrow \Delta, \varphi \mid \Pi\Rightarrow\Sigma \quad \Gamma\Rightarrow\Delta \mid \psi,
		\Pi \Rightarrow \Sigma } {\Gamma
		\Rightarrow \Delta \mid \varphi\rightarrow\psi, \Pi\Rightarrow\Sigma}$

	\egroup	
	\caption{Rules of propositional $ \bf K_3$}
	\label{fig::CalculusK3}
\end{figure}

Pavlovi\'{c} and Gratzl \cite{PavlovicGratzl23} formalise NFL by means of some nonstandard SC, where sequents are of the form: $\Gamma; \Pi\Rightarrow\Sigma; \Delta$. However we prefer to use for this aim bisequent calculus (BSC) which was already applied as a uniform basis for arbitrary three-valued logics in \cite{IndrzejczakPetrukhin}. Similar sequent system (but with a slightly different semantic interpretation) was applied also by Fjellstad \cite{Fjellstad1,Fjellstad2}.
Bisequents are ordered pairs of sequents $\Gamma\Rightarrow\Delta \mid \Pi\Rightarrow\Sigma$, where $\Gamma, \Delta, \Pi, \Sigma$ are
finite (possibly empty) multisets of formulae. The correspondence of bisequents to nonstandard sequents from \cite{PavlovicGratzl23} is indicated by using the same metavariables $\Gamma, \Delta, \Pi, \Sigma$ in both cases. $B$ and $S$ are used as metavariables for bisequents and sequents, respectively. 

In both systems of NFL a bisequent $\Gamma\Rightarrow\Delta \mid \Pi\Rightarrow\Sigma$ is axiomatic iff for some atomic formula $\varphi$ (including identities and $Et$), either $\varphi\in\Gamma\cap\Sigma$ or $\varphi\in\Gamma\cap\Delta$ or $\varphi\in\Pi\cap\Sigma$. 
Moreover, bisequents of the form $ \Ex t_1,\ldots,\Ex t_n, \Gamma\Rightarrow\Delta, P(t_1,\ldots,t_n) \mid P(t_1,\ldots,t_n),\Pi\Rightarrow\Sigma$ are also axiomatic.

As follows from \cite{IndrzejczakPetrukhin}, propositional connectives of $ \bf K_3$ are characterised by the rules from Fig. \ref{fig::CalculusK3}. The propositional connectives of $ \bf K_3^w$ are mostly as in $ \bf K_3$, however in four cases we have  three-premiss rules given in Fig. \ref{fig::CalculusK3w}. 
Fig. \ref{fig::CalculusQ} contains the BSC rules for the universal quantifier as well as the existence predicate in the spirit of the rules given in \cite{PavlovicGratzl23}.

\begin{figure}[t!]
	\centering
	\bgroup

\noindent {\footnotesize $(\mathord{\mid\wedge_w}\mathord{\Rightarrow})$ \ $\dfrac{\Gamma
		\Rightarrow \Delta \mid \varphi, \psi, \Pi\Rightarrow\Sigma \qquad \Gamma\Rightarrow\Delta, \varphi \mid \varphi,
		\Pi \Rightarrow \Sigma \qquad \Gamma\Rightarrow\Delta, \psi \mid \psi, \Pi\Rightarrow\Sigma} {\Gamma
		\Rightarrow \Delta \mid \varphi\wedge\psi, \Pi\Rightarrow\Sigma}$}

\1

\noindent {\footnotesize $(\mathord{\mid\Rightarrow}\mathord{\wedge_w})$\ $\dfrac{\Gamma
		\Rightarrow \Delta \mid \Pi\Rightarrow\Sigma, \varphi, \psi \qquad \varphi, \Gamma\Rightarrow\Delta \mid \Pi \Rightarrow \Sigma, \psi \qquad \psi, \Gamma\Rightarrow\Delta \mid \Pi\Rightarrow\Sigma, \varphi} {\Gamma
		\Rightarrow \Delta \mid \Pi\Rightarrow\Sigma, \varphi\wedge\psi}$}

\1

\noindent {\footnotesize $(\mathord{\Rightarrow}\mathord{\rightarrow_w}\mid)$ \ $\dfrac{\Gamma
		\Rightarrow \Delta, \psi \mid \varphi, \Pi\Rightarrow\Sigma \qquad \Gamma\Rightarrow\Delta, \varphi \mid \varphi,
		\Pi \Rightarrow \Sigma \qquad \Gamma\Rightarrow\Delta, \psi \mid \psi, \Pi\Rightarrow\Sigma} {\Gamma
		\Rightarrow \Delta, \varphi\rightarrow\psi \mid \Pi\Rightarrow\Sigma}$}

\1

\noindent {\footnotesize $(\mathord{\rightarrow_w}\mathord{\Rightarrow}\mid)$ \ $\dfrac{\varphi, \psi, \Gamma
		\Rightarrow \Delta \mid \Pi\Rightarrow\Sigma \qquad \Gamma\Rightarrow\Delta \mid \Pi \Rightarrow \Sigma, \varphi, \psi \qquad \psi, \Gamma\Rightarrow\Delta \mid \Pi\Rightarrow\Sigma, \varphi} {\varphi\rightarrow\psi, \Gamma
		\Rightarrow \Delta \mid \Pi\Rightarrow\Sigma}$}

\egroup	
\caption{Specific rules of propositional $ \bf K_3^w$}
\label{fig::CalculusK3w}
\end{figure}

\begin{figure}[t!]
	\centering
	\bgroup			

\begin{flushleft}
$(\mathord{\Rightarrow}\mathord{\forall\mid})$ \, $\dfrac{\Ex a,\Gamma
	\Rightarrow \Delta,\varphi[x/a] \mid S} {\Gamma
		\Rightarrow \Delta,\forall x \varphi \mid S}$\qquad	
$(\mathord{\mid}\mathord{\Rightarrow}\mathord{\forall})$ \, $\dfrac{\Ex a,\Gamma
	\Rightarrow \Delta\mid\Pi
	\Rightarrow \Sigma,\varphi[x/a]} {\Gamma
		\Rightarrow \Delta\mid\Pi
		\Rightarrow \Sigma,\forall x\varphi }$		
\end{flushleft}
\begin{flushleft}
$(\mathord{\forall}\mathord{\Rightarrow}\mathord{\mid})$ \, $\dfrac{\Ex b,\forall x \varphi,\varphi[x/b],\Gamma
	\Rightarrow \Delta \mid S} {\Ex b,\forall x \varphi,\Gamma
		\Rightarrow \Delta \mid S}$\quad
$(\mathord{\mid}\mathord{\forall}\mathord{\Rightarrow})$ \, $\dfrac{\Ex b,\Gamma
	\Rightarrow \Delta\mid\forall x\varphi,\varphi[x/b],\Pi
	\Rightarrow \Sigma} {\Ex b,\Gamma
		\Rightarrow \Delta\mid \forall x\varphi,\Pi
		\Rightarrow \Sigma}$			
\end{flushleft}

\begin{flushleft}
$(\Ex \mathord{\Rightarrow}\mathord{\mid})$ \, $\dfrac{\Ex t,P[t],\Gamma
	\Rightarrow \Delta \mid S} {P[t],\Gamma
		\Rightarrow \Delta \mid S}$\qquad
$(\mathord{\mid}\mathord{\Rightarrow}\Ex)$ \, $\dfrac{\Ex t,\Gamma
		\Rightarrow \Delta\mid\Pi
	\Rightarrow \Sigma,P[t]} {\Gamma
			\Rightarrow \Delta\mid\Pi
		\Rightarrow \Sigma,P[t]}$			
\end{flushleft}
\begin{flushleft}
$(\mathord{\mid}\Ex\mathord{\Rightarrow})$ \, $\dfrac{\Ex t_1,\ldots,\Ex t_n,P(t_1,\ldots,t_n),\Gamma
	\Rightarrow \Delta \mid \Pi\Rightarrow\Sigma} {\Ex t_1,\ldots,\Ex t_n,\Gamma
		\Rightarrow \Delta \mid P(t_1,\ldots,t_n),\Pi\Rightarrow\Sigma}$\qquad $ (\Ex Tr_1) $ $ \dfrac{\Ex t,\Gamma
		\Rightarrow \Delta \mid \Pi\Rightarrow\Sigma}{\Gamma
				\Rightarrow \Delta \mid \Ex t,\Pi\Rightarrow\Sigma} $\qquad
\end{flushleft}
\begin{flushleft}		
$(\mathord{\Rightarrow}\Ex\mathord{\mid})$ \, $\dfrac{\Ex t_1,\ldots,\Ex t_n,\Gamma
	\Rightarrow \Delta \mid \Pi\Rightarrow\Sigma,P(t_1,\ldots,t_n)} {\Ex t_1,\ldots,\Ex t_n,\Gamma
		\Rightarrow \Delta,P(t_1,\ldots,t_n) \mid \Pi\Rightarrow\Sigma}$\qquad			$ (\Ex Tr_2) $ $ \dfrac{\Gamma
		\Rightarrow \Delta \mid  \Pi\Rightarrow\Sigma,\Ex t}{\Gamma
				\Rightarrow \Delta,\Ex t \mid \Pi\Rightarrow\Sigma} $	
\end{flushleft}

\begin{flushleft}
where $ a $ is fresh and $b, b_i$ are arbitrary parameters. Both $ P[t] $  and $ P(t_1,\ldots,t_n) $ denote atoms or identities but not $ \Ex t $, moreover identities of the form $b=d$ are excluded since they are governed by rules from figure 5.  In $P[t]$ there is at least one occurrence of $t$ and there may be other terms; in $P(t_1,\ldots,t_n) $ there are no other terms.
\end{flushleft}

\egroup	
\caption{Rules for universal quantifier and existence predicate}
\label{fig::CalculusQ}
\end{figure}

Note that the applications of $(\mathord{\forall}\mathord{\Rightarrow}\mathord{\mid})$ and $(\mathord{\mid}\mathord{\forall}\mathord{\Rightarrow})$ are restricted to parameters as instantiated terms, no DD are allowed as instantiated terms! It saves us from losing the subformula property. 

In negative FL a standard identity is replaced by its weaker variant with restricted reflexivity $Et\rightarrow t=t$. In NFL we treat identity similarly so it behaves like other predicates and this is why the four rules for $E$ from Fig. \ref{fig::CalculusQ} apply also to identities. 

\begin{figure}[t!]
	\centering
	\bgroup 
	\begin{scriptsize}

	\begin{flushleft}
		{$(\mid\mathord{\Rightarrow}=)$ \ $\dfrac{\Gamma
				\Rightarrow \Delta \mid \Pi\Rightarrow\Sigma, t\approx s \qquad \Gamma\Rightarrow\Delta \mid 
				\Pi \Rightarrow \Sigma, A[x/t] \qquad A[x/s], \Gamma\Rightarrow\Delta \mid \Pi\Rightarrow\Sigma} {\Gamma
				\Rightarrow \Delta \mid \Pi\Rightarrow\Sigma}$}
	\end{flushleft}
	
	\begin{flushleft}

		$ (=\mathord{\Rightarrow}\mid) $ $ \dfrac{t=t, Et, \Gamma
			\Rightarrow \Delta \mid \Pi\Rightarrow\Sigma}{Et, \Gamma
			\Rightarrow \Delta \mid \Pi\Rightarrow\Sigma} $\qquad
		$(E\imath\mathord{\Rightarrow}\mid) $ $ \dfrac{a=d, Ed, \Gamma
			\Rightarrow \Delta \mid  \Pi\Rightarrow\Sigma}{Ed, \Gamma
			\Rightarrow \Delta\mid \Pi\Rightarrow\Sigma} $

	\end{flushleft}

\end{scriptsize}	
	where $A[x/t]$ is an atom, or identity or $Et$, $t\approx s$ denotes either $t=s$ or $s=t$, $a$ is a fresh parameter and $d$ is an arbitrary description.
	
	\egroup	
	\caption{Rules for identity}
	\label{fig::CalculusId}
\end{figure}

Notice that the rule $(\mid\mathord{\Rightarrow}=)$, required for proving the Leibniz Principle, is like the rule from \cite{IndrzejczakFreeDD}. We lose unrestricted subformula property but not in the essential way; in fact only terms $s, t$ and atoms $A$ are needed which are already present in the conclusion of this rule. It is of course possible to use other rules, with smaller branching factor and satisfying the subformula property, to obtain the same effect. However, in the presence of rules for DD with identities as principal formulas, it leads to serious troubles with proving cut admissibility. On the other hand, this rule avoids the problems and allows us to prove cut admissibility for reasonably small price. It is also possible to prepare a more refined version of the calculus, like in \cite{IndrzejczakRussellDD}, with the separation of rules for different kinds of identities, but this leads to the proliferation of rules and, because of space restrictions, we present here a simpler form of the calculus. A generalisation of quantifier rules to variants admitting DD as instantiated terms is easily provable since $\forall x\varphi, Ed\vdash \varphi[x/d]$ holds.

\begin{figure}[t!]
	\centering
	\bgroup

	\begin{scriptsize}

		\begin{flushleft}
			\begin{tabular}{ll}	
				$(\mathord{\imath \Rightarrow\mid 1})$ \ $\dfrac{\varphi[x/c], c=\imath x\varphi, Ec, \Gamma\Rightarrow\Delta \mid S}{c=\imath x\varphi, Ec, \Gamma\Rightarrow\Delta\mid S}$ &
				$(\mathord{\mid\imath\Rightarrow 1})$ \ $\dfrac{Ec, \Gamma
					\Rightarrow \Delta \mid \varphi[x/c], c=\imath x\varphi, \Pi\Rightarrow\Sigma}{Ec, \Gamma
					\Rightarrow \Delta \mid c=\imath x\varphi, \Pi\Rightarrow\Sigma }$\\[16pt]	
			\end{tabular}
		\end{flushleft}
		\begin{flushleft}
			$(\mathord{\imath \Rightarrow\mid 2})$ \, $\dfrac{c=\imath x\varphi, Eb, Ec, \Gamma
				\Rightarrow \Delta \mid \Pi\Rightarrow\Sigma, \varphi[x/b] \quad b=c, c=\imath x\varphi, Eb, Ec, \Gamma\Rightarrow\Delta \mid
				\Pi \Rightarrow \Sigma } {c=\imath x\varphi, Eb, Ec, \Gamma
				\Rightarrow \Delta \mid \Pi\Rightarrow\Sigma}$
		\end{flushleft}
		
		{\begin{flushleft}
				$(\mathord{\mid\imath\Rightarrow 2})$ \, $\dfrac{Eb, Ec, \Gamma
					\Rightarrow \Delta, \varphi[x/b] \mid c=\imath x\varphi, \Pi\Rightarrow\Sigma \quad Eb, Ec, \Gamma\Rightarrow\Delta \mid b=c, c=\imath x\varphi, 
					\Pi \Rightarrow \Sigma } {Eb, Ec, \Gamma
					\Rightarrow \Delta \mid c=\imath x\varphi, \Pi\Rightarrow\Sigma}$
		\end{flushleft}}
		
		\begin{flushleft}
			$(\mathord{\mid\Rightarrow\imath})$ \, $\dfrac{Ec, \Gamma
				\Rightarrow \Delta \mid \Pi\Rightarrow\Sigma, \varphi[x/c] \quad Ea, Ec,  \varphi[x/a], \Gamma\Rightarrow\Delta \mid
				\Pi \Rightarrow \Sigma, a=c} {Ec, \Gamma
				\Rightarrow \Delta \mid \Pi\Rightarrow\Sigma, c=\imath x\varphi, }$
		\end{flushleft}
		
		\begin{flushleft}
			$(\mathord{\Rightarrow \imath\mid})$ \, $\dfrac{Ec, \Gamma
				\Rightarrow \Delta, \varphi[x/c] \mid \Pi\Rightarrow\Sigma \quad Ea, Ec, \Gamma\Rightarrow\Delta, a=c \mid \varphi[x/a], 
				\Pi \Rightarrow \Sigma } {Ec, \Gamma
				\Rightarrow \Delta, c=\imath x\varphi \mid \Pi\Rightarrow\Sigma}$
		\end{flushleft}
		
	\end{scriptsize}
	
	where $a$ is a fresh parameter.
	\egroup	
	\caption{Rules for DD}
	\label{fig::CalculusDD}
\end{figure}

How BSC relates to the semantics from Section \ref{Logics}? 
Following \cite[p. 331]{IndrzejczakPetrukhin}, we give the subsequent definition. 
\begin{defn}
	$\models\Gamma\Rightarrow\Delta \mid \Pi\Rightarrow\Sigma$ iff 
	every valuation $\mathcal{V}$ satisfies $\Gamma\Rightarrow\Delta \mid \Pi\Rightarrow\Sigma$. The latter holds for $\mathcal{V}$ iff for some $\varphi$: either ($\varphi\in \Gamma$ and $\mathcal{V}(\varphi)\neq 1$) or ($\varphi\in \Delta$ and $\mathcal{V}(\varphi)=1$) or ($\varphi\in \Pi$ and $\mathcal{V}(\varphi)=0$) or ($\varphi\in \Sigma$ and $\mathcal{V}(\varphi)\neq 0$). Clearly $\not \models\Gamma\Rightarrow\Delta \mid \Pi\Rightarrow\Sigma$ iff for some $\mathcal{V}$, all elements of $\Gamma$ are true, all elements of $\Delta$ are either false or undefined, all elements of $\Pi$ are either true or undefined and all elements of $\Sigma$ are false. In this case we say that $\mathcal{V}$ falsifies this sequent.
\end{defn}

As follows from \cite{IndrzejczakPetrukhin}, all axiomatic bisequents are valid  and all the rules for connectives are both sound (validity-preserving) and invertible. The same holds for the rules from Fig. \ref{fig::CalculusQ} as follows from \cite{PavlovicGratzl23}. It may be extended to the new rules from Fig. \ref{fig::CalculusId} and \ref{fig::CalculusDD}, hence it holds:

\begin{theorem}[Soundness]\label{Soundness}
	For any bisequent $ B $, if $ \vdash B$, then $ \models B$.
\end{theorem}

\begin{proof}
By induction of the height of the derivation, using the fact that all the rules are sound. As an example, let us illustrate soundness of the rule  $(\mathord{\mid\Rightarrow\imath})$. The proof holds both for $ \bf K_3$ and $ \bf K_3^w$. Recall that $ a $ is fresh. 
\begin{enumerate}[$ (1) $]\itemsep=0pt
\item Suppose that $ \models Ec,\Gamma
	\Rightarrow \Delta \mid \Pi\Rightarrow\Sigma, \varphi[x/c]$.
\item Thus, for any valuation $ \val $, $ \val(Ec)\not=1 $, or for some $ \gamma\in\Gamma $, $ \val(\gamma)\not=1 $, or for some $ \delta\in\Delta $, $ \val(\delta)=1$, or for some $ \pi\in\Pi $, $ \val(\pi)=0 $, or for some $ \sigma\in\Sigma $, $ \val(\sigma)\not=0 $, or $ \val(\varphi[x/c])\not=0 $.
\item Suppose that $ \models Ea, \varphi[x/a],Ec, \Gamma\Rightarrow\Delta \mid
	\Pi \Rightarrow \Sigma, a=c$. 
\item Hence, for any valuation $ \val $, $ \val(Ea)\not=1 $, or $ \val(\varphi[x/a])\not=1 $, or $ \val(Ec)\not=1 $, or for some $ \gamma\in\Gamma $, $ \val(\gamma)\not=1 $, or for some $ \delta\in\Delta $, $ \val(\delta)=1$, or for some $ \pi\in\Pi $, $ \val(\pi)=0 $, or for some $ \sigma\in\Sigma $, $ \val(\sigma)\not=0 $, or $ \val(a=c)\not=0 $.
\item $ \not\models Ec,\Gamma
	\Rightarrow \Delta \mid \Pi\Rightarrow\Sigma, c=\imath x\varphi$.
\item Therefore, there is a valuation $ \val $ such that $ \val(Ec)=1 $, for any $ \gamma\in\Gamma $, $ \val(\gamma)=1 $, for any $ \delta\in\Delta $, $ \val(\delta)\not=1 $, for any $ \pi\in\Pi $, $ \val(\pi)\not=0 $, and for any $ \sigma\in\Sigma $, $ \val(\sigma)=0 $, $ \val(c=\imath x\varphi)=0 $.
\item Then we obtain:	\begin{enumerate}\itemsep=0pt
\item $ \val(\varphi[x/c])\not=0 $,
\item $ \val(Ea)\not=1 $, or $ \val(\varphi[x/a])\not=1 $, or $ \val(a=c)\not=0 $,
\item $ \val(Ec)=1 $, $ \val(c=\imath x\varphi)=0 $.
\end{enumerate}
\item By the properties of $ \mathcal{I}(=) $, we have $ c,\imath x\varphi \in\mathcal{I}(E)$.
\item $ \mathcal{I}(\imath x\varphi)\not=c $, that is $ \val(\varphi[x/c])\not=1 $ or $ \val(\varphi[x/b])\not=0 $, for some $ c\not=b\in\mathcal{I}(E) $.
\item Since $ a $ is fresh, $ \val(\varphi[x/c])\not=1 $ or \big($ \val(\varphi[x/a])\not=0 $ and $ c\not=a\in\mathcal{I}(E) $\big).
\item Suppose that $ \val(\varphi[x/c])\not=1 $. Thus, $ \val(\varphi[x/c])=\half $, since also $ \val(\varphi[x/c])\not=0 $. Since $ c\in\mathcal{I}(E) $, $ \val(\varphi[x/c])\not=\half $. Contradiction.
\item Suppose that $ \val(\varphi[x/a])\not=0 $ and $ c\not=a\in\mathcal{I}(E) $.
\item Thus, $ \val(Ea)=1 $ and $\val(a=c)\not=1$.
\item Hence, $ \val(\varphi[x/a])\not=1 $ or $ \val(a=c)\not=0 $.
\item If $ \val(\varphi[x/a])\not=1 $, then $ \val(\varphi[x/a])=\half $, since $ \val(\varphi[x/a])\not=0 $. Since $ a\in\mathcal{I}(E) $, $ \val(\varphi[x/a])\not=\half $. Contradiction.
\item If $ \val(a=c)\not=0 $, then $ \val(a=c)=\half $, since $ \val(a=c)\not=1 $. However, $ \val(a=c)\not=\half $, because $ a,c\in\mathcal{I}(E) $. Contradiction. 
\item Contradiction. Thus, $ \models Ec,\Gamma
	\Rightarrow \Delta \mid \Pi\Rightarrow\Sigma, c=\imath x\varphi$.
\end{enumerate}

\end{proof}

In fact, for the NFL without DD and identity, as defined by rules from Fig. \ref{fig::CalculusK3}, \ref{fig::CalculusK3w}, \ref{fig::CalculusQ}, also completeness was proved in \cite{PavlovicGratzl23}, and it may be extended in a straightforward way to cover identity, as was done for negative FL in \cite{PavlovicGratzl21}. For DD we already demonstrated that \eqref{L} is valid, so we restrict our considerations of adequacy to show that in (complete) BSC for NFL without DD we can prove interderivability of our rules for DD with both forms of \eqref{L}.

To do that we must restrict our general notion of provability in BSC. 
Although we defined satisifability and validity for arbitrary bisequents, in fact we are interested in the narrower notion of BSC proof, related to definition of consequence relation, as defined in Section \ref{Logics}. 
The definition of validity for bisequents shows that the entailment between $\Gamma$ and $\varphi$ is expressed in BSC as provability of $\Gamma\Rightarrow\varphi \mid\; \Rightarrow$. Moreover, although in propositional Kleene's logic the set of validities is empty it is not so in the first-order NFL based on it. In particular, both forms of \eqref{L} are valid and will be shown provable.

To save space in the proofs we will usually omit repetitions of formulae which are still active in premisses, according to rule schema, but which are not essential for the proof in question. The following three results hold:

\begin{lemma}[Substitution]\label{Substitution}
	If $\vdash_n\Gamma\Rightarrow\Delta \mid \Theta\Rightarrow\Lambda$ is derivable 
	(where $ n $ denotes derivability with height bounded by $ n $), then the sequent
	$\vdash_n\Gamma[b/c]\Rightarrow\Delta[b/c]\mid \Theta[b/c]\Rightarrow\Lambda[b/c]$ is likewise derivable.	
\end{lemma}
\begin{proof} Similar to the proof in \cite{PavlovicGratzl23}. Note that it is restricted to parameters. 
\end{proof}

\begin{lemma}
	
	$ \vdash\Ex t_1,\ldots,\Ex t_n, \Gamma\Rightarrow\Delta, \varphi(t_1,\ldots,t_n) \mid \varphi(t_1,\ldots,t_n),\Pi\Rightarrow\Sigma$, where $t_1, \ldots, t_n$ are all terms occuring in $\varphi$
	
\end{lemma}
\begin{proof} By induction on the complexity of $\varphi$. 	
\end{proof}

\begin{lemma}[Leibniz Law]
	For any formula $ \varphi $, it holds that 	$\vdash t_1 = t_2, \varphi[x/t_1] \Rightarrow \varphi[x/t_2] \mid S$.
\end{lemma}
\begin{proof} By induction on the complexity of $\varphi$. 
\end{proof}

The proof of \eqref{L2} in strong Kleene logic is as follows, where ($ E $) stands for provable $Ea\Rightarrow \varphi[x/a] \mid \varphi[x/a]\Rightarrow$:

{\normalsize \begin{prooftree}
		\EnableBpAbbreviations	
		\AXC{($ E $)}
		\AxiomC{$Eb\Rightarrow \varphi[x/b] \mid \varphi[x/b]\Rightarrow$}
		\RL{$(\mid\rightarrow\Rightarrow)$}
		\AxiomC{$Ea, Eb\Rightarrow a=b \mid a=b \Rightarrow$}
		\BinaryInfC{$Ea, Eb\Rightarrow a=b\mid \varphi[x/b], \varphi[x/b]\rightarrow a=b\Rightarrow$}
		\RL{$(\mid\forall\Rightarrow)$}
		\UnaryInfC{$Ea, Eb\Rightarrow a=b\mid \varphi[x/b], \forall x(\varphi\rightarrow a=x)\Rightarrow$}
		\RightLabel{$(\Rightarrow\imath\mid)$}
		\BinaryInfC{$Ea\Rightarrow a=\imath x\varphi\mid\varphi[x/a], \forall x(\varphi\rightarrow a=x)\Rightarrow$}
		\RightLabel{$(\mid\wedge\Rightarrow)$}
		\UnaryInfC{$Ea\Rightarrow a=\imath x\varphi\mid\varphi[x/a]\wedge\forall x(\varphi\rightarrow a=x)\Rightarrow$}
		\RightLabel{$(\Rightarrow\rightarrow\mid)$}
		\UnaryInfC{$Ea\Rightarrow \varphi[x/a]\wedge\forall x(\varphi\rightarrow a=x)\rightarrow a=\imath x\varphi\mid \Rightarrow$}
		\RightLabel{$(\Rightarrow\forall\mid)$}
		\UnaryInfC{$\Rightarrow \forall y(\varphi[x/y]\wedge\forall x(\varphi\rightarrow y=x)\rightarrow y=\imath x\varphi)\mid \Rightarrow$}
\end{prooftree}}

For \eqref{L1} let ($ E^\prime $) stand for the following proof:
\begin{prooftree}
	
	\EnableBpAbbreviations
	\AxiomC{$Ea\Rightarrow \varphi[x/a] \mid \varphi[x/a]\Rightarrow$}
	\RL{$(\mathord{\mid\imath\Rightarrow 1})$}
	\UnaryInfC{$Ea\Rightarrow \varphi[x/a] \mid a=\imath x\varphi\Rightarrow$}
\end{prooftree}

Then:
{\normalsize \begin{prooftree}
		\EnableBpAbbreviations
		\AXC{($ E^\prime $)}		
		\AxiomC{$Eb\Rightarrow \varphi[x/b] \mid \varphi[x/b]\Rightarrow$}
		\RL{$(\mathord{\mid\imath\Rightarrow 2})$}
		\AxiomC{$Ea, Eb\Rightarrow a=b \mid a=b \Rightarrow$}
		\BinaryInfC{$Ea, Eb\Rightarrow a=b \mid \varphi[x/b], a=\imath x\varphi\Rightarrow$}
		\RL{$(\mathord{\Rightarrow\rightarrow\mid})$}
		\UnaryInfC{$Ea, Eb\Rightarrow \varphi[x/b]\rightarrow a=b \mid a=\imath x\varphi\Rightarrow$}
		\RL{$(\mathord{\Rightarrow\forall\mid})$}
		\UnaryInfC{$Ea\Rightarrow\forall x(\varphi\rightarrow a=x)\mid a=\imath x\varphi \Rightarrow$}
		\RL{$(\mathord{\Rightarrow\wedge\mid})$}
		\BinaryInfC{$Ea\Rightarrow \varphi[x/a]\wedge\forall x(\varphi\rightarrow a=x)\mid a=\imath x\varphi\Rightarrow$}
		\RL{$(\mathord{\Rightarrow\rightarrow\mid})$}
		\UnaryInfC{$Ea\Rightarrow a=\imath x\varphi\rightarrow\varphi[x/a]\wedge\forall x(\varphi\rightarrow a=x)\mid \Rightarrow$}
		\RL{$(\mathord{\Rightarrow\forall\mid})$}
		\UnaryInfC{$\Rightarrow \forall y(y=\imath x\varphi\rightarrow\varphi[x/y]\wedge\forall x(\varphi\rightarrow y=x))\mid \Rightarrow$}
\end{prooftree}}

Proofs in BSC for weak Kleene's logic are more complex since the applications of $(\Rightarrow\rightarrow\mid)$ and $(\mid\wedge\Rightarrow)$ introduce additional premisses. However, in all respective premisses we obtain simply sequents of the form $Ea, \Gamma \Rightarrow \Delta,\psi(a)\mid\psi(a),\Pi\Rightarrow\Sigma$ or $Ea, Eb,\Gamma \Rightarrow \Delta,\psi(a,b)\mid\psi(a,b),\Pi\Rightarrow\Sigma$ which are provable.

All rules for DD are derivable if we use \eqref{L1} or \eqref{L2} 
 as additional axioms and use cuts which will be proved admissible in the next section. Here is an example. From both premisses of $(\Rightarrow\imath\mid)$ we obtain:
 \begin{footnotesize}
 	\begin{prooftree}
 		\AxiomC{$Ec, \Gamma\Rightarrow\Delta, \varphi[x/c]\mid \Pi\Rightarrow\Sigma$}
 		\AxiomC{$Ec, Ea, \Gamma\Rightarrow\Delta, c=a\mid \varphi[x/a], \Pi\Rightarrow\Sigma$}
 		\RightLabel{$(\Rightarrow\rightarrow\mid)$}
 		\UnaryInfC{$Ec, Ea, \Gamma\Rightarrow \Delta, \varphi[x/a]\rightarrow c=a\mid \Pi\Rightarrow\Sigma$}
 		\RightLabel{$(\Rightarrow\forall\mid)$}
 		\UnaryInfC{$Ec, \Gamma\Rightarrow \Delta, \forall x(\varphi\rightarrow c=x)\mid \Pi\Rightarrow\Sigma$}
 		\RightLabel{$(\mid\Rightarrow \wedge)$}
 		\BinaryInfC{$Ec, \Gamma\Rightarrow \Delta, \varphi[x/c]\wedge\forall x(\varphi\rightarrow c=x)\mid \Pi\Rightarrow\Sigma$}
 	\end{prooftree}	
 \end{footnotesize}		
 
 which by cut with $Ec, \varphi[x/c]\wedge\forall x(\varphi\rightarrow c=x)\Rightarrow c=\imath x\varphi \mid \Rightarrow$ yields the conclusion  of $(\Rightarrow\imath\mid)$. The latter bisequent is proved by cuts on the axiomatic sequent $(L^\leftarrow)$ i.e. $\Rightarrow \forall y(\varphi[x/y]\wedge\forall x(\varphi\rightarrow y=x)\rightarrow y=\imath x\varphi)\mid \Rightarrow$ with $Ec, \forall y(\varphi[x/y]\wedge\forall x(\varphi\rightarrow y=x)\rightarrow y=\imath x\varphi)\Rightarrow \varphi[x/c]\wedge\forall x(\varphi\rightarrow c=x)\rightarrow c=\imath x\varphi\mid \Rightarrow$ and $\varphi[x/c]\wedge\forall x(\varphi\rightarrow c=x)\rightarrow c=\imath x\varphi), \varphi[x/c]\wedge\forall x(\varphi\rightarrow c=x)\Rightarrow c=\imath x\varphi\mid \Rightarrow$ which are easily provable.

\section{Cut Admissibility}\label{Cut}

In order to show that BSC is cut-free we need to prove the following results:

\begin{lemma}[Generalisation of axioms]
For any formula $ \varphi $, the following bisequents are derivable: \begin{enumerate}[$(1)$]\itemsep=0pt
\item $ \varphi,\Gamma\Rightarrow\Delta\mid\Pi\Rightarrow\Sigma,\varphi $,
\item $ \varphi,\Gamma\Rightarrow\Delta,\varphi\mid\Pi\Rightarrow\Sigma $, 
\item $ \Gamma\Rightarrow\Delta\mid\varphi,\Pi\Rightarrow\Sigma,\varphi $.
\item $ \Ex t_1, \hdots, \Ex t_n, \Gamma\Rightarrow\Delta, \varphi[t_1, \hdots, t_n]\mid\varphi[t_1, \hdots, t_n],\Pi\Rightarrow\Sigma$.
\end{enumerate}
\end{lemma}
\begin{proof}
 By induction on the complexity of $ \varphi $. 
\end{proof}

Here we present proofs of the admissibility of cut with all forms of cuts from \cite{PavlovicGratzl23}. But first we demonstrate the admissibility of structural rules, like in \cite{PavlovicGratzl23}.

\begin{lemma}[Admissibility of structural rules]
	Structural rules from Fig. \ref{fig::CalculusStr} are height-preserving admissible.
\end{lemma}
\begin{proof}
	By induction on the height of the derivation. 
\end{proof}

In fact, the proof of admissibility of contraction presupposes the following:

\begin{lemma}[Invertibility]
All the rules of the bisequent calculi in question are height-preserving invertible.
\end{lemma}
\begin{proof}
 By induction on the height of the derivation, using Lemma \ref{Substitution}. 
\end{proof}

We also need the following:

\begin{lemma}[Transfer]
	The following rules are height-preserving admissible:
	
	\1
	
	$(LTr)$ \, $\dfrac{\Gamma
		\Rightarrow \Delta \mid \varphi, \Pi\Rightarrow\Sigma}{\varphi, \Gamma\Rightarrow\Delta \mid \Pi\Rightarrow\Sigma}$ \hspace{1cm}
	$(RTr)$ \, $\dfrac{\Gamma
		\Rightarrow \Delta, \varphi \mid\Pi\Rightarrow\Sigma}{\Gamma\Rightarrow\Delta \mid \Pi\Rightarrow\Sigma, \varphi}$

\end{lemma}
\begin{proof}
	Straightforward extension of the proof by induction on the height of the derivation from \cite{PavlovicGratzl23}. 
\end{proof}

\begin{figure}[t!]
	\centering
	\bgroup		
\begin{scriptsize}

\begin{flushleft}
	\begin{tabular}{llll}	
		$(\mathord{W}\mathord{\Rightarrow\mid})$ \ $\dfrac{\Gamma
			\Rightarrow \Delta \mid S}{\varphi, \Gamma \Rightarrow \Delta \mid S}$ &
		$(\mathord{\Rightarrow}\mathord{W\mid})$ \ $\dfrac{\Gamma
			\Rightarrow \Delta \mid S}{\Gamma
			\Rightarrow \Delta, \varphi \mid S}$ &
		$(\mathord{\mid W}\mathord{\Rightarrow})$ \ $\dfrac{S \mid \Pi\Rightarrow\Sigma}{S\mid \varphi, \Pi\Rightarrow\Sigma}$ &
		$(\mathord{\mid\Rightarrow}\mathord{W})$ \ $\dfrac{S \mid \Pi\Rightarrow\Sigma}{S \mid \Pi\Rightarrow\Sigma, \varphi}$\\[16pt]	
	\end{tabular}
\end{flushleft}

\begin{flushleft}
	\begin{tabular}{llll}	
		$(\mathord{C}\mathord{\Rightarrow\mid})$ \ $\dfrac{\varphi,\varphi,\Gamma
			\Rightarrow \Delta \mid S}{\varphi, \Gamma \Rightarrow \Delta \mid S}$ &
		$(\mathord{\Rightarrow}\mathord{C\mid})$ \ $\dfrac{\Gamma
			\Rightarrow \Delta, \varphi, \varphi \mid S}{\Gamma
			\Rightarrow \Delta, \varphi \mid S}$ &		
		$(\mathord{\mid C}\mathord{\Rightarrow})$ \ $\dfrac{S \mid\varphi,\varphi, \Pi\Rightarrow\Sigma}{S\mid \varphi, \Pi\Rightarrow\Sigma}$ &
		$(\mathord{\mid\Rightarrow}\mathord{C})$ \ $\dfrac{S \mid \Pi\Rightarrow\Sigma, \varphi, \varphi}{S \mid \Pi\Rightarrow\Sigma, \varphi}$\\[16pt]	
	\end{tabular}
\end{flushleft}

\end{scriptsize}

\egroup	
\caption{Admissible structural rules}
\label{fig::CalculusStr}
\end{figure}

In bisequent framework, we have  several cut rules\footnote{Notice that we considered just two cut rules on the propositional level \cite{IndrzejczakPetrukhin}. However, the first-order case requires more options, which are actually the same as Pavlovi\'{c} and Gratzl have \cite{PavlovicGratzl23}. They follow the strategy of Fjellstad \cite{Fjellstad2} in this respect.} listed in Fig. \ref{fig::CalculusCut}.

\begin{figure}[t!]
	\centering
	\bgroup
	\begin{scriptsize}			

\begin{center}
	(E-Cut) $ \dfrac{\Gamma\Rightarrow\Delta\mid \Lambda\Rightarrow\Theta,\Ex t\quad\Ex t,\Pi \Rightarrow\Sigma \mid\Xi\Rightarrow\Omega}{\Gamma,\Pi\Rightarrow\Delta,\Sigma \mid \Lambda,\Xi\Rightarrow\Theta,\Omega} $
\end{center}

\begin{center}
(L-Cut) $ \dfrac{\Ex t_1,\ldots,\Ex t_n,\Gamma\Rightarrow\Delta\mid \Lambda\Rightarrow\Theta,P(t_1,\ldots,t_n)\quad\Ex t_1,\ldots,\Ex t_n,P(t_1,\ldots,t_n),\Pi \Rightarrow\Sigma \mid\Xi\Rightarrow\Omega}{\Ex t_1,\ldots,\Ex t_n\Gamma,\Pi\Rightarrow\Delta,\Sigma \mid \Lambda,\Xi\Rightarrow\Theta,\Omega} $
\end{center}

\begin{center}
	(O-Cut) $ \dfrac{\Gamma\Rightarrow\Delta, \varphi\mid\Lambda\Rightarrow\Theta\quad \varphi,\Pi \Rightarrow\Sigma \mid\Xi\Rightarrow\Omega}{\Gamma,\Pi\Rightarrow\Delta,\Sigma \mid \Lambda,\Xi\Rightarrow\Theta,\Omega} $
\end{center}
\begin{center} 
	(I-Cut) $  \dfrac{\Gamma\Rightarrow\Delta\mid\Lambda\Rightarrow\Theta, \varphi\quad\Pi \Rightarrow\Sigma \mid\varphi,\Xi\Rightarrow\Omega}{\Gamma,\Pi\Rightarrow\Delta,\Sigma \mid \Lambda,\Xi\Rightarrow\Theta,\Omega} $
\end{center}
\begin{center} 
	(R-Cut) $ \dfrac{\Gamma\Rightarrow\Delta\mid \varphi,\Lambda\Rightarrow\Theta\quad\Pi \Rightarrow\Sigma,\varphi \mid\Xi\Rightarrow\Omega}{\Gamma,\Pi\Rightarrow\Delta,\Sigma \mid \Lambda,\Xi\Rightarrow\Theta,\Omega} $
\end{center}

\begin{center} 
(3-Cut)	$ \dfrac{\varphi,\Gamma_1\Rightarrow\Delta_1\mid \Lambda_1\Rightarrow\Theta_1\quad\Gamma_2\Rightarrow\Delta_2 \mid\Lambda_2\Rightarrow\Theta_2,\varphi\quad \Gamma_3\Rightarrow\Delta_3,\varphi \mid\varphi,\Lambda_3\Rightarrow\Theta_3}{\Gamma_1,\Gamma_2,\Gamma_3\Rightarrow\Delta_1,\Delta_2,\Delta_3 \mid \Lambda_1,\Lambda_2,\Lambda_3\Rightarrow\Theta_1,\Theta_2,\Theta_3} $
\end{center}
\end{scriptsize}

\egroup	
\caption{Cut rules}
\label{fig::CalculusCut}
\end{figure}

\begin{theorem}[Cut admissibility]
The rules (E-Cut), (L-Cut), (O-Cut), (I-Cut), (R-Cut), (3-Cut) are admissible.	
\end{theorem}

\begin{proof}
 Admissibility of $(E-Cut)$ and $(L-Cut)$ is proved first in \cite{PavlovicGratzl23}, and their proof applies to our system with no changes. The remaining variants of cut are
proved in \cite{PavlovicGratzl23} simultaneously, by double induction on the complexity of the cut formula 
and on the height of the cut (the sum of heights of premises of cut). The extension of their proof to cover our rules is in some cases straightforward so we consider only the most complex situation with triple cut:

{\tiny \begin{center} 
	$ \dfrac{\Gamma_1\Rightarrow\Delta_1\mid \Lambda_1\Rightarrow\Theta_1, c=\imath x\varphi\quad c=\imath x\varphi, \Gamma_2\Rightarrow\Delta_2 \mid\Lambda_2\Rightarrow\Theta_2\quad Ec, \Gamma_3\Rightarrow\Delta_3,c=\imath x\varphi \mid c=\imath x\varphi,\Lambda_3\Rightarrow\Theta_3}{Ec, \Gamma_1,\Gamma_2,\Gamma_3\Rightarrow\Delta_1,\Delta_2,\Delta_3 \mid \Lambda_1,\Lambda_2,\Lambda_3\Rightarrow\Theta_1,\Theta_2,\Theta_3} $
\end{center}}

for simplicity let us refer to the premisses of this cut as $P_l, P_m, P_r$.
Since the middle premiss $P_m$ may be derived either by $(\imath\Rightarrow\mid1)$ or $(\imath\Rightarrow\mid2)$ and the rightmost premiss $P_r$ may be derived either by $(\Rightarrow\imath\mid)$ (on the left identity) or $(\mid\imath\Rightarrow1)$ or $(\mid\imath\Rightarrow2)$ (on the right identity) we have six subcases in total.
The leftmost premiss $P_l$ is derivable only by $(\mid\Rightarrow\imath)$ hence we have always:
\begin{prooftree}
	\AxiomC{$D_1$}
	\UnaryInfC{$\Gamma_1\Rightarrow\Delta_1\mid \Lambda_1\Rightarrow\Theta_1, \varphi[x/c]$}
	\AxiomC{$D_2$}
	\UnaryInfC{$Ea, \varphi[x/a], \Gamma_1\Rightarrow\Delta_1,\mid \Lambda_1\Rightarrow\Theta_1, c=a$}
	\RightLabel{$(\mid\Rightarrow\imath)$}
	\BinaryInfC{$\Gamma_1\Rightarrow\Delta_1\mid\Lambda_1\Rightarrow\Theta_1, c=\imath x\varphi$}
\end{prooftree}

1. and let the middle premiss be derived as follows:
\begin{prooftree}	
	\AxiomC{$D_3$}
	\UnaryInfC{$c=\imath x\varphi, \varphi[x/c], \Gamma_2\Rightarrow\Delta_2,\mid \Lambda_2\Rightarrow\Theta_2$}
	\RightLabel{$(\imath\Rightarrow\mid)$}
	\UnaryInfC{$c=\imath x\varphi, \Gamma_2\Rightarrow\Delta_2\mid\Lambda_2\Rightarrow\Theta_2 $}
\end{prooftree}

1.1. and the rightmost one:
\begin{prooftree}
	\AxiomC{$D_4$}
	\UnaryInfC{$Ec, \Gamma_3\Rightarrow\Delta_3,c=\imath x\varphi \mid c=\imath x\varphi,\varphi[x/c],\Lambda_3\Rightarrow\Theta_3$}
	\RightLabel{$(\mid\imath\Rightarrow1)$}
	\UnaryInfC{$Ec, \Gamma_3\Rightarrow\Delta_3,c=\imath x\varphi \mid c=\imath x\varphi,\Lambda_3\Rightarrow\Theta_3$}
\end{prooftree}

let $ A_1 $ stands for the following:
\begin{prooftree}
	\AxiomC{$D_1$}
	\UnaryInfC{$\Gamma_1\Rightarrow\Delta_1\mid \Lambda_1\Rightarrow\Theta_1, \varphi[x/c]$}
\end{prooftree}

we transform the proof as follows:
{\small \begin{prooftree}
\AxiomC{$A_1$}
	\AxiomC{$P_l$}
	\AxiomC{$P_m$}	
	\AxiomC{$D_4$}
	\UnaryInfC{$Ec, \Gamma_3\Rightarrow\Delta_3,c=\imath x\varphi \mid c=\imath x\varphi,\varphi[x/c],\Lambda_3\Rightarrow\Theta_3$}
	\RightLabel{(3-Cut)}
	\TrinaryInfC{$Ec, \Gamma_1,\Gamma_2,\Gamma_3\Rightarrow\Delta_1,\Delta_2,\Delta_3 \mid\varphi[x/c], \Lambda_1,\Lambda_2,\Lambda_3\Rightarrow\Theta_1,\Theta_2,\Theta_3$}
	\RightLabel{(I-Cut)}
	\BinaryInfC{$Ec, \Gamma_1,\Gamma_1,\Gamma_2,\Gamma_3\Rightarrow\Delta_1,\Delta_1\Delta_2,\Delta_3\mid\Lambda_1,\Lambda_1,\Lambda_2,\Lambda_3\Rightarrow\Theta_1,\Theta_1,\Theta_2,\Theta_3$}
	\RightLabel{$(C)$}
	\UnaryInfC{$Ec, \Gamma_1,\Gamma_2,\Gamma_3\Rightarrow\Delta_1,\Delta_2,\Delta_3\mid\Lambda_1,\Lambda_2,\Lambda_3\Rightarrow\Theta_1,\Theta_2,\Theta_3$}
\end{prooftree}}

where $(3-Cut)$ is admissible by induction on the height and $(I-Cut)$ by induction on the degree.

1.2. $P_r$ is derived as follows:
\begin{prooftree}
\AxiomC{$A_2$}
\AxiomC{$A_3$}
	\RightLabel{$(\mid\imath\Rightarrow 2)$}
	\BinaryInfC{$Ec, \Gamma_3\Rightarrow\Delta_3,c=\imath x\varphi \mid c=\imath x\varphi,\Lambda_3\Rightarrow\Theta_3$}
\end{prooftree}

where $ A_2 $ stands for
\begin{prooftree}
	\AxiomC{$D_4$}
	\UnaryInfC{$Ec, \Gamma_3\Rightarrow\Delta_3,c=\imath x\varphi, \varphi[x/b] \mid c=\imath x\varphi,\Lambda_3\Rightarrow\Theta_3$}
\end{prooftree}

and $ A_3 $ stands for
\begin{prooftree}
	\AxiomC{$D_5$}
	\UnaryInfC{$Ec, \Gamma_3\Rightarrow\Delta_3,c=\imath x\varphi \mid c=\imath x\varphi,c=b,\Lambda_3\Rightarrow\Theta_3$}
\end{prooftree}

note that also $Eb$ must occur in $\Gamma_3$. We perform:

$(3-Cut)$ on $P_l, P_m$ and $Ec, \Gamma_3\Rightarrow\Delta_3,c=\imath x\varphi, \varphi[x/b] \mid c=\imath x\varphi,\Lambda_3\Rightarrow\Theta_3$ yield $S_4 := Ec, \Gamma_1,\Gamma_2,\Gamma_3\Rightarrow\Delta_1,\Delta_2,\Delta_3, \varphi[x/b]\mid\Lambda_1,\Lambda_2,\Lambda_3\Rightarrow\Theta_1,\Theta_2,\Theta_3$

$(3-Cut)$ on $P_l, P_m$ and $Ec, \Gamma_3\Rightarrow\Delta_3,c=\imath x\varphi\mid c=\imath x\varphi,c=b,\Lambda_3\Rightarrow\Theta_3$ yield $S_5 := Ec, \Gamma_1,\Gamma_2,\Gamma_3\Rightarrow\Delta_1,\Delta_2,\Delta_3, \mid c=b,\Lambda_1,\Lambda_2,\Lambda_3\Rightarrow\Theta_1,\Theta_2,\Theta_3$

both cuts are admissible by induction on the height. By substitution lemma on $D_2$ we obtain a proof of $S_2 := Eb, \varphi[x/b], \Gamma_1\Rightarrow\Delta_1,\mid \Lambda_1\Rightarrow\Theta_1, c=b$. We finish with:

{\small \begin{prooftree}
	\AxiomC{$S_5$}
	\AxiomC{$S_4$}	
	\AxiomC{$S_2$}
	\RightLabel{(O-Cut)$ (C) $}
	\BinaryInfC{$Ec, \Gamma_1,\Gamma_2,\Gamma_3\Rightarrow\Delta_1,\Delta_2,\Delta_3 \mid c=b, \Lambda_1,\Lambda_2,\Lambda_3\Rightarrow\Theta_1,\Theta_2,\Theta_3$}
	\RightLabel{(I-Cut)$(C)$}
	\BinaryInfC{$Ec, \Gamma_1,\Gamma_2,\Gamma_3\Rightarrow\Delta_1,\Delta_2,\Delta_3\mid\Lambda_1,\Lambda_2,\Lambda_3\Rightarrow\Theta_1,\Theta_2,\Theta_3$}
\end{prooftree}}

where both cuts are admissible by induction on the degree.

1.3. $P_r$ is derived as follows:

\begin{prooftree}
\AxiomC{$A_4$}
\AxiomC{$A_5$}
	\RightLabel{$(\Rightarrow\imath\mid)$}
	\BinaryInfC{$Ec,\Gamma_3\Rightarrow\Delta_3,c=\imath x\varphi \mid c=\imath x\varphi,\Lambda_3\Rightarrow\Theta_3$}
\end{prooftree}

where $ A_4 $ stands for
\begin{prooftree}
	\AxiomC{$D_4$}
	\UnaryInfC{$Ec, \Gamma_3\Rightarrow\Delta_3,\varphi[x/c] \mid c=\imath x\varphi,\Lambda_3\Rightarrow\Theta_3$}
\end{prooftree}

and $ A_5 $ stands for
\begin{prooftree}
	\AxiomC{$D_5$}
	\UnaryInfC{$Ea',Ec, \Gamma_3\Rightarrow\Delta_3,c=a'\mid c=\imath x\varphi,\varphi[x/a'],\Lambda_3\Rightarrow\Theta_3$}
\end{prooftree}

where $a'$ is eigenvariable different from $a$. We perform:

{\small \begin{prooftree}
	\AxiomC{$\Gamma_1\Rightarrow\Delta_1\mid\Lambda_1\Rightarrow\Theta_1, c=\imath x\varphi$}
	\AxiomC{$D_4$}
	\UnaryInfC{$Ec, \Gamma_3\Rightarrow\Delta_3, \varphi[x/c] \mid c=\imath x\varphi,\Lambda_3\Rightarrow\Theta_3$}
	\RightLabel{(I-Cut)}
	\BinaryInfC{$Ec, \Gamma_1, \Gamma_3\Rightarrow\Delta_1, \Delta_3,\varphi[x/c] \mid \Lambda_1,\Lambda_3\Rightarrow\Theta_1,\Theta_3$}
\end{prooftree}}

and

{\footnotesize \begin{prooftree}
	
	\AxiomC{$\Gamma_1\Rightarrow\Delta_1\mid\Lambda_1\Rightarrow\Theta_1, c=\imath x\varphi$}
\AxiomC{$A_6$}
	\AxiomC{$Ec,\Gamma_3\Rightarrow\Delta_3,c=\imath x\varphi \mid c=\imath x\varphi,\Lambda_3\Rightarrow\Theta_3$}
	\TrinaryInfC{$Ec, \varphi[x/c],\Gamma_1,\Gamma_2,\Gamma_3\Rightarrow\Delta_1,\Delta_2,\Delta_3 \mid \Lambda_1,\Lambda_2,\Lambda_3\Rightarrow\Theta_1,\Theta_2,\Theta_3\qquad\qquad\qquad(3-Cut)$}
\end{prooftree}}

where $ A_6 $ stands for
\begin{prooftree}
	\AxiomC{$D_3$}
	\UnaryInfC{$c=\imath x\varphi, \varphi[x/c], \Gamma_2\Rightarrow\Delta_2,\mid \Lambda_2\Rightarrow\Theta_2$}
\end{prooftree}

Both cuts are admissible by induction on the height and the resulting sequents by $(O-Cut)$ lead to desired result.

\1

2. Now let the middle premiss be derived as follows:

{\footnotesize \begin{prooftree}
\AxiomC{$A_7$}
\AxiomC{$A_8$}
	\RightLabel{$(\imath\Rightarrow\mid)$}
	\BinaryInfC{$Eb,c=\imath x\varphi,\Gamma_2\Rightarrow\Delta_2 \mid\Lambda_2\Rightarrow\Theta_2$}
\end{prooftree}}

where $ A_7 $ stands for
\begin{prooftree}
	\AxiomC{$D_3$}
	\UnaryInfC{$Eb, c=\imath x\varphi,\Gamma_2\Rightarrow\Delta_2\mid \Lambda_2\Rightarrow\Theta_2,\varphi[x/b]$}
\end{prooftree}

and $ A_8 $ stands for
\begin{prooftree}
	\AxiomC{$D_4$}
	\UnaryInfC{$Eb, c=b, c=\imath x\varphi,\Gamma_2\Rightarrow\Delta_2\mid\Lambda_2\Rightarrow\Theta_2$}
\end{prooftree}

2.1. and the right premiss by:

\begin{prooftree}
	\AxiomC{$D_5$}
	\UnaryInfC{$Ec, \Gamma_3\Rightarrow\Delta_3,c=\imath x\varphi \mid c=\imath x\varphi,\varphi[x/c],\Lambda_3\Rightarrow\Theta_3$}
	\RightLabel{$(\mid\imath\Rightarrow 1)$}
	\UnaryInfC{$Ec, \Gamma_3\Rightarrow\Delta_3,c=\imath x\varphi \mid c=\imath x\varphi,\Lambda_3\Rightarrow\Theta_3$}
\end{prooftree}

and the sequent resulting by $(3-Cut)$ is:

$Eb, Ec, \Gamma_1,\Gamma_2,\Gamma_3\Rightarrow\Delta_1,\Delta_2,\Delta_3 \mid \Lambda_1,\Lambda_2,\Lambda_3\Rightarrow\Theta_1,\Theta_2,\Theta_3$

we transform the proof as follows:

{\small \begin{prooftree}
\AxiomC{$A_9$}
	\AxiomC{$P_l$}
	\AxiomC{$P_m$}	
	\AxiomC{$D_5$}
	\UnaryInfC{$Ec, \Gamma_3\Rightarrow\Delta_3,c=\imath x\varphi \mid c=\imath x\varphi,\varphi[x/c],\Lambda_3\Rightarrow\Theta_3$}
	\RightLabel{(3-Cut)}
	\TrinaryInfC{$Eb, Ec, \Gamma_1,\Gamma_2,\Gamma_3\Rightarrow\Delta_1,\Delta_2,\Delta_3 \mid\varphi[x/c], \Lambda_1,\Lambda_2,\Lambda_3\Rightarrow\Theta_1,\Theta_2,\Theta_3$}
	\RightLabel{(I-Cut)}
	\BinaryInfC{$Eb,Ec, \Gamma_1,\Gamma_1,\Gamma_2,\Gamma_3\Rightarrow\Delta_1,\Delta_1\Delta_2,\Delta_3\mid\Lambda_1,\Lambda_1,\Lambda_2,\Lambda_3\Rightarrow\Theta_1,\Theta_1,\Theta_2,\Theta_3$}
	\RightLabel{$(C)$}
	\UnaryInfC{$Eb,Ec, \Gamma_1,\Gamma_2,\Gamma_3\Rightarrow\Delta_1,\Delta_2,\Delta_3\mid\Lambda_1,\Lambda_2,\Lambda_3\Rightarrow\Theta_1,\Theta_2,\Theta_3$}
\end{prooftree}}

where $ A_9 $ stands for
\begin{prooftree}
	\AxiomC{$D_1$}
	\UnaryInfC{$\Gamma_1\Rightarrow\Delta_1\mid \Lambda_1\Rightarrow\Theta_1, \varphi[x/c]$}
\end{prooftree}

2.2. If the right premiss is derived:

\begin{prooftree}
\AxiomC{$A_{10}$}
\AxiomC{$A_{11}$}
	\RightLabel{$(\mid\imath\Rightarrow 2)$}
	\BinaryInfC{$Ec, \Gamma_3\Rightarrow\Delta_3,c=\imath x\varphi \mid c=\imath x\varphi,\Lambda_3\Rightarrow\Theta_3$}
\end{prooftree}

where $ A_{10} $ stands for
\begin{prooftree}
	\AxiomC{$D_5$}
	\UnaryInfC{$Ec, \Gamma_3\Rightarrow\Delta_3,c=\imath x\varphi, \varphi[x/b'] \mid c=\imath x\varphi,\Lambda_3\Rightarrow\Theta_3$}
\end{prooftree}

where $ A_{11} $ stands for
\begin{prooftree}
	\AxiomC{$D_6$}
	\UnaryInfC{$Ec, \Gamma_3\Rightarrow\Delta_3,c=\imath x\varphi \mid c=\imath x\varphi,c=b',\Lambda_3\Rightarrow\Theta_3$}
\end{prooftree}

note that also $Eb'$ must occur in $\Gamma_3$ ($b'$ distinct from $b$).

$(3-Cut)$ on $P_l, P_m$ and $D_5$ yields $S_7 :=  Eb, Ec,\Gamma_1,\Gamma_2, \Gamma_3\Rightarrow\Delta_1,\Delta_2,\Delta_3,\allowbreak\varphi[x/b'] \mid \Lambda_1,\Lambda_2,\Lambda_3\Rightarrow\Theta_1,\Theta_2,\Theta_3$

$(3-Cut)$ on $P_l, P_m$ and $D_6$ yields $S_8 :=  Eb, Ec,\Gamma_1,\Gamma_2, \Gamma_3\Rightarrow\Delta_1,\Delta_2,\Delta_3\mid c=b', \Lambda_1,\Lambda_2,\Lambda_3\Rightarrow\Theta_1,\Theta_2,\Theta_3$

with cuts admissible by induction on the height.

By substitution lemma on $D_2$ we get $S_9 := Eb', \varphi[x/b'], \Gamma_1\Rightarrow\Delta_1,\mid \Lambda_1\Rightarrow\Theta_1, c=b'$

$(I-Cut)$ and $(O-Cut)$ made on these three sequents yield the desired result.

2.3. Finally let the right premiss be obtained by:

\begin{prooftree}
\AxiomC{$A_{12}$}
\AxiomC{$A_{13}$}
	\RightLabel{$(\Rightarrow\imath\mid)$}
	\BinaryInfC{$Ec,\Gamma_3\Rightarrow\Delta_3,c=\imath x\varphi \mid c=\imath x\varphi,\Lambda_3\Rightarrow\Theta_3$}
\end{prooftree}

where $ A_{12} $ stands for
\begin{prooftree}
	\AxiomC{$D_5$}
	\UnaryInfC{$Ec, \Gamma_3\Rightarrow\Delta_3,\varphi[x/c] \mid c=\imath x\varphi,\Lambda_3\Rightarrow\Theta_3$}
\end{prooftree}

and $ A_{13} $ stands for
\begin{prooftree}
	\AxiomC{$D_6$}
	\UnaryInfC{$Ea',Ec, \Gamma_3\Rightarrow\Delta_3,c=a'\mid c=\imath x\varphi,\varphi[x/a'],\Lambda_3\Rightarrow\Theta_3$}
\end{prooftree}

where $a'$ is eigenvariable different from $a$.

By substitution lemma on $D_6$ we get $S := Eb,Ec,\Gamma_3\Rightarrow\Delta_3,c=b\mid c=\imath x\varphi,\varphi[x/b], \Lambda_3\Rightarrow\Theta_3$ (the proof has the same height). $S$ by $(I-Cut)$ with $P_l$ gives $S_1 := Eb,Ec,\Gamma_1, \Gamma_3\Rightarrow\Delta_1,\Delta_3,c=b\mid \varphi[x/b], \Lambda_1\Lambda_3\Rightarrow\Theta_1,\Theta_3$

$(3-Cut)$ on $P_l, D_3$ and $P_r$ yields $S_2 :=  Eb, Ec,\Gamma_1,\Gamma_2, \Gamma_3\Rightarrow\Delta_1,\Delta_2,\Delta_3 \mid \Lambda_1,\Lambda_2,\Lambda_3\Rightarrow\Theta_1,\Theta_2,\Theta_3,\varphi[x/b]$

$(3-Cut)$ on $P_l, D_4$ and $P_r$ yields $S_3 :=  Eb, Ec, c=b,\Gamma_1,\Gamma_2, \Gamma_3\Rightarrow\Delta_1,\Delta_2,\allowbreak\Delta_3\mid \Lambda_1,\Lambda_2,\Lambda_3\Rightarrow\Theta_1,\Theta_2,\Theta_3$

All these cuts are admissible by induction on the height.

$S_1, S_2, S_3$ by $(I-Cut)$ and $(O-Cut)$ yield the desired result. 

\end{proof}

\section{Conclusion}\label{Conclusion}
To the best of our knowledge this paper offers the first proof-theoretic study of NFL with identity and DD. The most important task for the presented variant of NFL is to find a satisfactory solution which is cut free and does not restrict the treatment of DD to terms which do not admit nesting of other DD inside. Such terms like `the owner of the biggest diamond' should be dealt with in a direct way.
Moreover, \eqref{L} characterises only the behaviour of proper DD so the present system provides the weakest theory of DD.   
Since, as we mentioned, NFL seems to be the most natural framework for developing a satisfactory theory of improper DD, the most important task for the future is to develop stronger theories of DD.

Our intention is also to explore alternative approaches to NFL (see \cite{Lehmann,Woodruff}) based on different definitions of propositional connectives, viable options concerning the treatment of existence predicate, identity, and alternative interpretations of quantifiers. Moreover, taking into account that careless treatment of DD leads to contradictions, it is an important task to explore its behaviour as founded on paraconsistent logics, in the framework of BSC.

The next step would be to explore the problems of proof search and automatisation for NFL and related systems. It seems that the completeness proof of \cite{PavlovicGratzl23}, allowing for building countermodels, can be extended to cover identity and DD by using techniques applied in \cite{IndZaw2} or \cite{IndNils}. Preparation of analytic tableau systems for this kind of logics and their implementation would be another welcome output, when the theoretical foundations will be firmly established.

\paragraph{Acknowledgements.} We would like to thank anonymous reviewers for their valuable comments and suggestions.

\end{document}